\documentclass[3p,times]{elsarticle}

\usepackage{amssymb,amsmath,amsthm}

\newtheorem{definition}{Definition}
\newtheorem{thm}{Theorem}
\newtheorem{proposition}[thm]{Proposition}
\newtheorem{lem}[thm]{Lemma}
\newtheorem{cor}[thm]{Corollary}
\newtheorem{remark}[thm]{Remark}
\newtheorem{notation}[thm]{Notation}
\def\ZZ{{\mathbb{Z}}}
\def\RR{{\mathbb{R}}}
\def\CC{{\mathbb{C}}}
\newcommand{\apici}[1]{\lq\lq{#1}\rq\rq}

\begin{document}

\begin{frontmatter}

\cortext[cor1]{Corresponding author. Fax: +39-090-393502.}

\title{A formula for the number of $(n-2)$-gap in digital $n$-objects}
\tnotetext[gnsaga]{This research was supported by P.R.I.N, P.R.A.
and  I.N.D.A.M. (G.N.S.A.G.A.).}

\author[unime]{Angelo MAIMONE
} \ead{angelo.maimone@unime.org}

\author[unime]{Giorgio NORDO\corref{cor1}
} \ead{giorgio.nordo@unime.org}
\address[unime]{Dipartimento di Matematica, Universit\`{a} degli Studi di Messina, Contrada Papardo, Salita Sperone 31, 98166, Sant'Agata, Messina.}

\begin{abstract}
We provide a formula that expresses the number of $(n-2)$-gaps of a
generic digital $n$-object. Such a formula has the advantage to involve only a few simple intrinsic parameters of the object and it is obtained by using a combinatorial  technic based on incidence structure and on the notion of free cells.
This approach seems suitable as a model for an automatic computation, and also allow us to find some expressions for the maximum number of $i$-cells that bound or are bounded by a fixed $j$-cell.
\end{abstract}


\end{frontmatter}

\section{Introduction}
With the word  {\lq\lq{gap}\rq\rq}  in Digital Geometry we mean some basic portion of a digital object that
a discrete ray can cross without intersecting any voxel of the object itself.
Since such a notion is strictly connected with some applications in the field of Computer graphics
(e.g. the rendering of a 3D image by the ray-tracing technique), many papers (see for example \cite{BMN arxiv}, \cite{BMMK Nevada},  \cite{BMN gap Berlino}, \cite{brimkov_maimone_nordo}, and \cite{Genus}) concerned the study of $0$- and $1$-gaps of $3$-dimensional objects and of some of their topological invariant such as dimension and genus (i.e. the degree of connectedness of a digital image).
Recently (see \cite{gap3D}), we have found a formula for expressing the number of $1$-gaps of a digital $3$-object by means of
the number of its free cells of dimension $1$ and $2$.
During the submission process of that paper, the anonymous referee raised to our attention the existence of another recent and more general  formula presented in \cite{B} which gives the number of a generical $(n-2)$-gaps of any digital $n$-object.
Unfortunately, such formula involves some parameters (the number of $(n-2)$-blocks and of $n$-,  $(n-1)$- and $(n-2)$- cells) that are non-intrinsic or that can not be easily obtained by the geometrical knowledge of the object.
For such a reason, in the present paper, we propose a generalization of the formula obtained in \cite{gap3D} that  allow us to express the number of $(n-2)$-gaps
using only two basic parameters, that is the number of free $(n-2)$- and $(n-1)$-cells of the object itself.
Although we prove the equivalence between these two formulas, the latter approach seems simpler and more suitable as a model for an automatic computation.
\\
In order to obtain our formula, we adopt a combinatorial  technic based on the notion of incidence structure, which also allow us to find a couple of interesting expressions for the maximum number of $i$-cells that bound or are bounded by a fixed $j$-cell.
\par
In the next section we recall and formalize some basic notions and notations
of digital geometry.
In Section \ref{sez:back}, we introduce the notions of tandem and gap, and we give some elementary facts about them.
In Section \ref{sez:main res}, we prove some propositions concerning, in particular, the number of $(n-1)$-cells  of the boundary of a digital object that are bounded by a given $(n-2)$-cell satisfying some particular condition, and we use such results to obtain our main formula for the number of $(n-2)$-gaps. Finally, in Section \ref{sez:concl}, we resume the goal of the paper and we give some suggestions for other future researches.

\section{Preliminaries}\label{sez:preliminari}
Throughout this paper we use the \textit{grid cell model} for
representing digital objects, and we adopt the terminology from
\cite{klette-rosenfeld} and \cite{Kovalevsky}.
\par
Let $x=(x_1,\ldots x_n)$ be a point of $\ZZ^n$, $\theta \in
\{-1,0,1\}^n$ be an $n$-word over the alphabet $\{-1,0,1\}$, and
$i\in\{1,\ldots n\}$. We define $i$-cell related to $x$ and
$\theta$, and we denote it by $ e = (x,\theta)$, the Cartesian
product, in a certain fixed order, of $n-i$ singletons $\left\{ x_j
\pm \frac{1}{2} \right\}$ by $i$ closed sets $\left[x_j -
\frac{1}{2}, x_j + \frac{1}{2}\right]$, i.e. we set
\[
e = (x,\theta) = \prod_{j=1}^n \left[ x_j  + \frac 1 2 \theta_j - \frac 1 2 [\theta_j=0], x_j  + \frac 1 2 \theta_j + \frac 1 2 [\theta_j=0]  \right],
\]
where $[\bullet]$ denotes the Iverson bracket \cite{Knuth}. The
word $\theta$ is called the \textit{direction} of the cell
$(x,\theta)$ related to the point $x$.
\\
Let us note that an $i$-cell can be related to different point $x
\in \ZZ^n$, and, once we have fixed it, can be related to different
direction. So, when we talk generically about $i$-cell, we mean one
of its possible representation.

The dimension of a cell $e = (x, \theta)$, denoted by  $\dim(e)=i$,  is the number of non-trivial interval of its product representation, i.e. the number of null components of its direction $\theta$.
Thus, $\dim(e) = \sum_{j = 1}^ n [\theta_j = 0]$ or, equivalently, $\dim(e)= n - \theta \cdot \theta$.
So, $e$ is an $i$-cell if and only if it  has dimension $i$.

We denote by $\CC_n^{(i)}$ the set of all $i$-cells of $\RR^n$ and
by $\CC_n$ the set of all cells defined in $\RR^n$, i.e. we set
$\CC_n=\bigcup_{j=0}^n \CC_n^{(j)}$. An $n$-cell of $\CC_n$ is also
called an $n$-voxel. So, for convenience, an $n$-voxel is denoted by
$v$, while we use other lower case letter (usually $e$) to denote
cells of lower dimension. A finite collection $D$ of $n$-voxels is a
digital $n$-object. For any $i=0,\ldots, n$, we denote by $C_{i}(D)$
the set of all $i$-cells of the object $D$, that is $D \cap
\CC_n^{(i)}$, and by $c_i(D)$ (or simply by $c_i$ if no confusion
arise) its cardinality $|C_i(D)|$.

\begin{definition}\label{Def: Definizione di centro di una cella}
Let $e = (x,\theta)$ be an $i$-cell. The center of $e$ is
defined by $cnt(e) = x + \frac 1 2 \theta$.
\end{definition}

\begin{remark}\label{Not:Su_espr_cell}
Let us note that for a cell $e = (x,\theta)$, we have  $cnt(e) = x$
if and only if $\dim(e)=n$. Moreover, thanks to Definition  \ref{Def:
Definizione di centro di una cella}, an $i$-cell related to $x$ and
$\theta$ can be shortly represented in the following way:
\[ e= \prod_{j=1}^n \left[ cnt(e)_j - \frac 1 2 [\theta_j=0], cnt(e)_j + \frac 1 2 [\theta_j=0]  \right]. \]
\end{remark}

\begin{definition}\label{Def:Cella duale}
Let $e=(x,\theta)$ be an $i$-cell related to the point $x$ and to the direction
$\theta$. We define dual $e'$  of $e$,  the cell represented by the  following cartesian product:
\[ e'=\prod_{j=1}^n \left[ cnt(e)_j - \frac 1 2 [\theta_j \ne 0], cnt(e)_j + \frac 1 2 [\theta_j \ne 0]  \right]. \]
\end{definition}

By the above expression and the definition of dimension of a cell, we have that the dimension of the dual  $e'$ of a cell $e= (x, \theta)$ coincides with the number of non-null components of the direction $\theta$, that is  $\dim(e') = \sum_{j=1}^n [\theta_j \ne 0]$.
Consequently, the dual $e'$ of an $i$-cell $e$ is an $(n-i)$-cell.

\begin{definition}\label{Def:Duale di un insieme}
Let $D$ be a digital object. The dual $D'$ of $D$ is the set of all dual cells $e'$, with $e \in D$.
\end{definition}

We say that two $n$-cells $v_1$, $v_2$ are $i$-adjacent ($ i =
0,1,\ldots, n-1$) if $ v_1\,\ne\, v_2 $ and there exists at least an
$i$-cell $\overline{e}$  such that $ \overline{e} \subseteq v_1 \cap
v_2$, that is if they are distinct and share at least an $i$-cell.
Two $n$-cells $v_1$, $v_2$ are \textit{strictly} $ i $-adjacent, if
they are $i$-adjacent but not $j$-adjacent, for any $j>i$, that is
if $v_1\cap v_2 \in \CC_n^{(i)}$. The set of all $n$-cells that are
$i$-adjacent to a given $n$-voxel $v$ is denoted by $A_i(v)$ and
called the $i$-\textit{adjacent neighborhoods} of $v$. Two cells
$v_1,v_2 \in \CC_n$ are \textit{incident} each other,  and we write $e_1 I e_2$, if $e_1
\subseteq e_2$ or $e_2 \subseteq e_1$.

\begin{definition}\label{Def: relazione di limitatezza}
Let $e_1,e_2 \in \CC_n$. We say that $e_1$ bounds $e_2$ (or that
$e_2$ is bounded by $e_1$), and we write $e_1 < e_2$, if $e_1 I e_2$
and $\dim(e_1) < \dim(e_2)$. The relation $<$ is called bounding
relation.
\end{definition}

\begin{definition}\label{Def:cella semplice}
Let $e$ be an $i$-cell of a digital $n$-object $D$ (with $i=0,\ldots
n-1$). We say that $e$ is simple if $e$ bounds one and only one
$n$-cell.
\end{definition}

\begin{definition}\label{Def:dualita_tra_insiemi}
Let $D$ and $G$ be two finite subsets of $\CC_n$. We say that $D$
and $G$ form a \textit{dual pair}  iff there exists a bijection
$\varphi \colon D \to G$ that inverts the bounded relation, that
is for any couple $e,f \in D$, if $e < f$ then $ \varphi(f) <
\varphi(e)$, and for any $e \in D$, $\dim(\varphi(e)) = n - \dim(e)$.
\end{definition}

\begin{proposition}
Let $D$ be a digital $n$-object and $D'$ its dual. Then $D$ and $D'$ form a dual pair.
\end{proposition}
\begin{proof}
Let us consider the mapping $\varphi \colon D \to D'$ that associates to each cell $e = (x,\theta) \in D$ its dual  $\varphi(e) = e'$.
Since, by Remark \ref{Not:Su_espr_cell} and Definition \ref{Def:Cella duale}, both $e$ and $e'$ are uniquely determinated by the point $x$ and the direction $\theta$, it is clear that $\varphi$ is a bijection.
\\
By a basic property of the Iverson notation, for every cell $e= (x, \theta)$, we have that
\[ \dim(\varphi(e)) = \dim(e') = \sum_{j=1}^n [\theta_j \ne 0] = \sum_{j=1}^n \left( 1 -  [\theta_j = 0] \right)
= n - \sum_{j=1}^n [\theta_j = 0] = n - \dim(e) . \]
Moreover, $\varphi$ inverts the bounding relation $<$ over $\CC_n$. Indeed, for every couple of cells $e = (x, \theta)$ and $f = (y, \psi)$ in $D$ such that $e < f$, without loss of generality,  we have that $ e \subseteq f$ and $\dim(e) < \dim(f)$. Thus, by Remark \ref{Not:Su_espr_cell}, we get
\[ \prod_{j=1}^n \left[ cnt(e)_j - \frac 1 2 [\theta_j=0], cnt(e)_j + \frac 1 2 [\theta_j=0]  \right] \subseteq \prod_{j=1}^n \left[ cnt(f)_j - \frac 1 2 [\psi_j=0], cnt(f)_j + \frac 1 2 [\psi_j=0]  \right] . \]
Hence, for every $j = 1, \ldots, n$, we have
\[ cnt(f)_j - \frac 1 2 [\psi_j=0] \le  cnt(e)_j - \frac 1 2 [\theta_j=0] \le  cnt(e)_j + \frac 1 2 [\theta_j=0] \le cnt(f)_j + \frac 1 2 [\psi_j=0]. \]
and so, we obtain
\begin{align*}
cnt(e)_j - \frac 1 2 [\theta_j \ne 0]  &
=  cnt(e)_j - \frac 1 2 \left( 1 - [\theta_j=0]\right)
=   cnt(e)_j  + \frac 1 2 [\theta_j=0]  - \frac 1 2  \le cnt(f)_j + \frac 1 2 [\psi_j=0]  - \frac 1 2 \\
 & = cnt(f)_j - \frac 1 2 [\psi_j \ne 0]
\le cnt(f)_j + \frac 1 2 [\psi_j \ne 0]
= cnt(f)_j + \frac 1 2 \left(  1-[\psi_j=0]  \right)  \\
& = cnt(f)_j  - \frac 1 2[\psi_j=0]   + \frac 1 2
\le  cnt(e)_j - \frac 1 2 [\theta_j=0]  + \frac 1 2
 = cnt(e)_j + \frac 1 2 [\theta_j \ne 0],
\end{align*}
which implies
\[ \prod_{j=1}^n \left[ cnt(f)_j - \frac 1 2 [\psi_j \ne 0], cnt(f)_j + \frac 1 2 [\psi_j \ne 0]  \right] \subseteq \prod_{j=1}^n \left[ cnt(e)_j - \frac 1 2 [\theta_j \ne 0], cnt(e)_j + \frac 1 2 [\theta_j \ne 0]  \right] . \]
Thus, $f' \subseteq e'$, i.e.  $\varphi(f) \subseteq \varphi(e)$.
Finally, since $\dim(e) < \dim(f)$, we have
$\dim(\varphi(f)) = n - \dim(f) < n - \dim(e) = \dim (\varphi(e)) $
and so $\varphi(f) < \varphi (e)$.
\end{proof}

\begin{definition}
An incidence structure (see \cite{Design_Theory})  is a triple $(V,\mathcal{B}, \mathcal I)$
where $V$ and $\mathcal{B}$ are any two disjoint sets and $\mathcal
I$ is a binary relation between $V$ and $\mathcal B$, that is
$\mathcal I \subseteq V \times \mathcal B$. The elements of $V$ are
called points, those of $\mathcal{B}$ blocks. Instead of $(p, B) \in
\mathcal I$, we simply write $p \mathcal I B$ and say that
\apici{the point $p$ lies on the block $B$} or \apici{$p$ and $B$
are incident}.
\end{definition}

If $p$ is any point of $V$, we denote by $(p)$ the set of all blocks
incident to $p$, i.e. $(p)=\{B\in \mathcal{B} \colon p \mathcal I
B\}$. Similarly, if $B$ is any block of $\mathcal B$,  we denote by
$(B)$ the set of all points incident to $B$, i.e. $(B)=\{p \in V
\colon p \mathcal I B \}$. For a point $p$, the number $r_p = |(p)|$
is called the degree of $p$, and similarly, for a block $B$, $k_B =
|(B)|$ is the degree of $B$.

We remind the following fundamental proposition of incidence
structures.
\begin{proposition}\label{Pro:relazione fondamentale delle strutture di incidenza}
Let $(V,\mathcal{B}, \mathcal I)$ be an incidence structure. We have
\begin{equation}\label{Eq: Equazione fondamentale delle strutture di incidenza}
    \sum_{p \in V} r_p = \sum_{B \in \mathcal{B}} k_B,
\end{equation}
where $r_p$ and $k_B$ are the degrees of any point $p \in V$ and any
block $B \in \mathcal B$, respectively.
\end{proposition}

\section{Theoretical Backgrounds}\label{sez:back}
In \cite{BMN gap Berlino} and \cite{gap3D}, a constructive
definition of gap for a digital object $D$ in spaces of dimensions
$2$ and $3$ was proposed, and a relation between the number of such
a gaps and the numbers of free cells was found.
\par
In order to generalize those results for the $n$-dimensional space,
we need to introduce some definitions and to make some considerations.

\begin{definition}\label{Def:blocchi and L-block}
Let $e$ be an $i$-cell (with $0\le i \le n-1$) of $\CC_n$.
Then:
\begin{enumerate}[(1)]
\item An
$i$-block centered on $e$ is the union of all the $n$-voxels
bounded by $e$, i.e. $B_i(e)= \bigcup\{ v \in \CC_{n}^{(n)} \colon e
< v \}$.
  \item An $L$-block centered on $e$ is an $(n-2)$-block centered on $e$ from
  which we take away one of its four  $n$-cells,
  that is $L(e) = B_{n-2} ( e ) \setminus \{ v \}$,
  where $v \in C_n(B_{n-2} ( e ) )$.
\end{enumerate}
\end{definition}

\begin{remark}\label{Not:Numero di nVoxel che compone un blocco e un L blocco}
Let us note that, for any $i$-cell $e$, $B_i(e)$ is the union of
exactly $2^{n-i}$ $n$-voxels,  $e \in B_i(e)$, and that an $L$-block
is exactly composed of three $n$-voxels.
\end{remark}

\begin{definition}\label{Def:tandemn}
Let $v_1$, $v_2$ be two $n$-voxels of a digital object $D$, and $e$
be an $i$-cell ($i=0,\ldots, n-1$). We say that $t_i = \{v_1,v_2\}$
forms an $i$-tandem of $D$ over $e$ if $D \cap B_i(e)= \{v_1,v_2\}
$, $v_1$ and $v_2$ are strictly $i$-adjacent and $v_1 \cap v_2 = e$.
\end{definition}

\begin{definition}\label{Def:gapn}
Let $D$ be a digital $n$-object and $e$ be an $i$-cell (with $i =
0,\ldots,n-2 $). We say that $D$ has an $i$-gap over $e$ if there
exists an $ i $-block $B_i(e)$ such that $B_i(e) \setminus D$ is an
$i$-tandem over $e$. The cell $e$ is called $i$-hub of the related
$i$-gap. Moreover, we denote by $g_i(D)$ (or simply by $g_i$ if no
confusion arises) the number of $i$-gap of $D$.
\end{definition}
 Examples of gaps for
$3$D case are given in Figure \ref{Fig:gap}.

\begin{figure}
  \centering
  \includegraphics{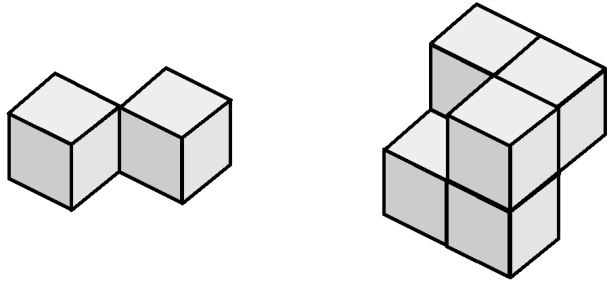}\\
  \caption{Configurations of  $1$- and $0$-gaps in $\CC_3$.  }\label{Fig:gap}
\end{figure}

\begin{proposition}\label{Pro:equivalence nmenodue tunnel}
A digital $n$-object $D$ has an  $(n-2)$-gap over an $(n-2)$-hub $e$
iff there exist two $n$-voxels $v_1$ and $v_2$ such that:
\begin{enumerate}[1)]
\item $e < v_1$ and $e < v_2$;
\item $v_1 \in A_{n-2}(v_2) \setminus A_{n-1}(v_2)$;
\item $A_{n-1}(v_1) \cap A_{n-1}(v_2) \cap D =\emptyset$.
\end{enumerate}
\end{proposition}
\begin{proof}
Let us suppose that $D$ has an $(n-2)$-gap over an $(n-2)$-hub $e$.
Then there exists an $(n-2)$-block $B = B_{n-2}(e)$ such that $B
\setminus D$ is an $(n-2)$-tandem over $e$. Hence $B \setminus D$ is
composed of two strictly $(n-2)$-adjacent $n$-voxel, let us say
$v_1, v_2$, and $v_1 \cap v_2 = e$. This implies that $e \subset
v_1$ and $e \subset v_2$, and so $e < v_1$ and $e < v_2$.
\\
Now, let us suppose that $v_1 \notin  A_{n-2}(v_2) \setminus
A_{n-1}(v_2)$. Then it should be $v_1 \notin A_{n-2}( v_2 )$ or $v_1
\in A_{n-1}(v_2)$. Both expressions  lead to a contradiction, since
$v_1$ and $v_2$  are strictly $(n-2)$-adjacent.
\\
Finally, let us suppose that $A_{n-1}(v_1) \cap A_{n-1}(v_2) \cap D
\not =\emptyset$. Then it should exists an $n$-voxel $v_3 \in D$
such that $ v_3 \in A_{n-1}(v_1) $ and $ v_3 \in A_{n-1}(v_2) $.
Hence $\{v_1,v_2,v_3\}$ forms an $L$-block. A contradiction since
$v_1$ and $v_2$ are strictly $(n-2)$-adjacent.
\par
Conversely, let us suppose that conditions $1)$, $2)$, and $3)$
hold, and, by contradiction, that for any $(n-2)$-cell $e \in D$, $E
= B_{n-2}(e) \setminus D$ is not an $(n-2)$-tandem over $e$. Then
$E$ is either an $i$-block ($i = n-2, n-1$) or an $L$-block whose
facts contradict our hypothesis.
\end{proof}

\begin{definition}\label{Def:FreeCell}
An $i$-cell $e$  (with $i=0, \ldots, n-1$) of a digital $n$-object $D$ is
free iff $B_i(e) \nsubseteq D$.
\end{definition}

For any $i=0,\ldots,n-1$, we denote by $C_i^*(D)$ (respectively by
$C_i'(D)$) the set of all free (respectively non-free) $i$-cells of
the object $D$. Moreover, we denote by $c_i^*(D)$ (or simply by
$c_i^*$) the number of free $i$-cells of $D$, and by $c_i'(D)$ (or
simply by $c_i'$) the number of non-free cells. It is evident that
$\{C_i^*(D),C_i'(D)  \}$ forms a partition of $C_i(D)$ and that
$c_i= c_i^* + c_i'$.

\begin{definition}\label{Def:Border}
The $i$-border ($i=1,\ldots,n-1$) $bd_i(D)$ of a digital $n$-object
$D$ is the set of all its $i$-cells  such that $B_i(e)$ intersects
both $D$ and $\CC_n \setminus D$. The union of all $i$-borders ($ 0
\le i\le n-1$) is called border of $D$ and denoted by $bd(D)$.
\end{definition}

An immediate consequence of Definitions \ref{Def:FreeCell} and
\ref{Def:Border} is given by the following proposition.
\begin{proposition}
An $i$-cell $e$ ($i = 0, \ldots, n-1$) of a digital object $D$ is free iff $e \in bd(D)$.
\end{proposition}

\begin{remark}\label{Rem:1}
The border $bd(D)$ of a digital $n$-object is composed of the set of
all free cells of $D$. Moreover, $c'_i$ coincides with the number of
all $i$-blocks $B_i(e)$ such that $B_i(e) \subseteq D$.
\end{remark}

\section{Main Results}\label{sez:main res}

\begin{definition}\label{Def: Fiore}
Let $e$ be an $i$-cells of $\CC_n$. The $j$-flower of $e$ ($i<j\le
n$) is the set of cells $F_j(e)$ constituted by all $j$-cells that
are bounded by $e$, that is we set $F_j(e)=\{c \in \CC_n^{(j)}
\colon e < c\}$. The cell $e$ is called the center of the flower,
while an element of $F_j(e)$ is called a $j$-petal (or simply petal
if confusion does not arise) of the $j$-flower $F_j(e)$.
\end{definition}

Let us note that Definition \ref{Def: Fiore} is a generalization of
the notion of $i$-block given in Definition \ref{Def:blocchi and
L-block}. Indeed an $i$-block centered on an $i$-cell $e$ can be
considered like the $n$-flower of $e$.

\begin{notation}
Let $i,j$ be two natural number such that $0 \le i < j$. We denote
by $ c_{i \rightarrow j} $ the maximum number of $i$-cells of
$\CC_n$ that bound a $j$-cell. Moreover, we denote by $c_{i
\leftarrow j}$ the maximum number of $j$-cell of $\CC_n$ that are
bounded by an $i$-cell.
\end{notation}
Let us note that, for any $0 \le i < j$, $c_{i \leftarrow j}$
represents the number of $j$-petal of the $j$-flower $F_j(e)$, where
$e$ is a cell of dimension $i$.

\begin{proposition}\label{Pro:numero massimo di i-celle che limitano una data j-cella}
For any $i,j \in \mathbb{N}$ such that $0 \le i <j$, it is
\[
c_{i \rightarrow j} = 2^{j-i} \binom{j}{i}.
\]
\end{proposition}
\begin{proof}
Since a $j$-cell of $\CC_n$ can be considered like an hypercube of
dimension $j$, the number $c_{i \rightarrow j}$ corresponds with the
number of $i$-faces of this hypercube which is $2^{j-i}
\binom{j}{i}$ (see, for example, \cite{libroPolitopi}).
\end{proof}

\begin{proposition}\label{Pro:numero massimo di i-celle che sono limitate da una j-cella fissata}
For any $i,j \in \mathbb{N}$ such that $0 \le i <j$, it is
\[
c_{i \leftarrow j} = 2^{j-i}\binom{n-i}{j-i}.
\]
\end{proposition}
\begin{proof}
Let $e$ be an $i$-cell of $\CC_n$, and let $F_j(e)$ be the related
$j$-flower. The dual $\Phi'$ of $\Phi = F_j(e) \cup \{e\}$ is an
incidence structure $(V, \mathcal B, \mathcal I)$, where $V= \{ p'
\colon p \in F_j(e) \}$, $\mathcal B = \{ e' \}$ and $\mathcal I$ is
the dual relation of the bounding relation $<$. Moreover, we have
$\dim(e') = n - i$ and $\dim(p')=n-j$. Hence, up to a bijection,
$\Phi'$ is the set composed of the $(n-i)$-cell $e'$ and by all the
possible $(n-j)$-cells which bound $e'$ . It follows that the maximum number $c_{i \leftarrow j}$ of
$j$-cells that are bounded by a given $i$-cell coincides with the
maximum number of $(n-j)$-cells that bound an $(n-i)$-cell, that is,
by Proposition \ref{Pro:numero massimo di i-celle che limitano una
data j-cella},
\[ c_{i \leftarrow j} = c_{n-j \rightarrow n-i} = 2^{n-i-n+j}\binom{n-i}{n-j} = 2^{j-i}\binom{n-i}{j-i}. \]
\end{proof}

\begin{lem}\label{Lem:come trovare il numero di n-1celle}
Let $D$ be a digital $n$-object. Then
\[ c_{n-1} = 2n c_n - c_{n-1}'. \]
\end{lem}
\begin{proof}
Let us consider the set
\[ F = \bigcup_{v \in   C_n(D)} \{ (e,v) \colon e \in C_{n-1}(D), e < v \}.\] It is evident that $ \big|F \big| =  \Big|\{ (e,v) \colon  e \in C_{n-1} (D), e < v \} \Big|   \cdot  \Big| C_{n}(D)\Big|  = c_{n-1 \rightarrow n} \cdot c_n = 2 n c_n $.
Let us set $F^* = F \cap (C_{n-1}^*(D) \times C_{n}(D))$ and $F' = F \cap (C_{n-1}'(D) \times C_{n}(D))$. The map $\phi \colon F^* \to C_{n-1}^*(D)$, defined by $\phi(e,v)=e$, is a bijection. In fact, besides being evidently surjective, it is also injective, since, if by contradiction there were two distinct pairs $(e,v_1)$ and $(e,v_2) \in F^*$ associated to $e$, then $B_{n-1}(e)= \{v_1,v_2\}$ should be an $(n-1)$-block contained in $D$. This contradicts the fact that the $(n-1)$-cell $e$ is free. Thus $|F^*| = |C_{n-1}^*(D)| = c_{n-1}^*$.
\\
On the other hand, $\big|F' \big| = \Big| \displaystyle \bigcup_{v
\in   C_n (D)} \{ (e,v) \colon e \in C'_{n-1}(D), e < v \} \Big| =
\Big|\displaystyle \bigcup_{e \in   C'_{n-1} (D)} \{ (e,v) \colon v
\in C_{n}(D), e < v \} \Big| = \Big|\{ (e,v) \colon  v \in C_{n}
(D), e < v \} \Big| \cdot  \Big| C'_{n-1}(D)\Big|  = c_{n-1
\leftarrow n} \cdot c'_{n-1} = 2 c'_{n-1}$. Since $\{ F^*, F' \}$ is
a partition of $F$, we finally have that $| F | = |F^*| + |F'|$,
that is $   2n c_n = c_{n-1}^* + 2 c_{n-1}' = c_{n-1} -  c_{n-1}' +
2  c_{n-1}' = c_{n-1} + c_{n-1}'$, and then the thesis.
\end{proof}

\begin{notation}
Let $e$ be an $i$-cell of a digital $n$-object $D$, and $0 \le i <
j$. We denote by $b_j(e,D)$ (or simply by $b_j(e)$ if no confusion
arises) the number of $j$-cells of $bd(D)$ that are bounded by $e$.
\end{notation}
Let us note that if $e$ is a non-free $i$-cell, then $b_j(e) = 0$.

\begin{definition}{\label{Def: no gap}}
A free $i$-cell of a digital $n$-object that is not an $i$-hub is
called $i$-nub.
\end{definition}

\begin{notation}\label{Not:Nota sulla cardinalità di H e N}
For any $i=0, \ldots, n-1$, we denote by $\mathcal H_{i}(D)$ and by
$\mathcal N_{i}(D)$ (or simply by $\mathcal H_{i}$ and by $\mathcal
N_{i}$ if no confusion arises) the sets of $i$-hubs and $i$-nubs of
$D$, respectively. We have $|\mathcal H_{i}| = g_{i}$ and $|
\mathcal N_{i} | = c^*_i - g_i$.
\end{notation}

\begin{figure}
  \centering
  \includegraphics{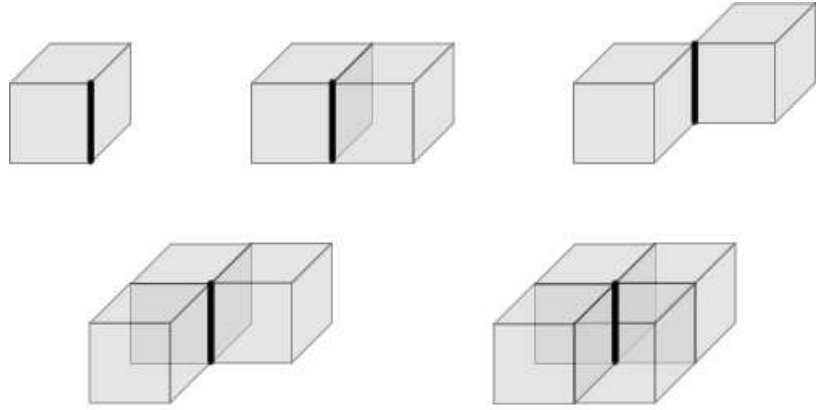}\\
  \caption{The five possible cases for the set $V = \{ v \in \CC_n^{(n)}
\colon e<v \}$ in $3$D case. The black thick segment represents the edge $e$.}\label{Fig:Ese1}
\end{figure}

We are interested in classifying all the
possible configurations  of $n$-voxels
bounded by an $(n-2)$-cell $e$.

\begin{lem}
Let $e$ be an $(n-2)$-cell of $\CC_n$, and $V = \{ v \in \CC_n^{(n)}
\colon e<v \}$  be the set of $n$-voxels bounded by $e$. Then one
and only one of the following five cases occurs (See Figure
\ref{Fig:Ese1} for an example for $3$D case):
\begin{itemize}
  \item $V$ is a singleton and $e$ is a simple cell;
  \item $V$ is an $(n-1)$-block centered on an $(n-1)$-cell that is bounded by $e$;
  \item $V$ is $(n-2)$-gap and $e$ is its $(n-2)$-hub;
  \item $V$ is an $L$-block and $e$ is its center;
  \item $V$ is an $(n-2)$-block and $e$ is its center.
\end{itemize}
\end{lem}
\begin{proof}
By Definition \ref{Def:blocchi and L-block}(1),  the largest set of
$n$-voxels bounded by $e$ is the $(n-2)$-block centered on $e$.
Moreover, by Remark \ref{Not:Numero di nVoxel che compone un blocco
e un L blocco},  $c_n(B_{n-2}(e)) = 4$. Hence, the number $c_n(V)$
of $n$-voxels of $V$ have to be between one and four and,
up to symmetries, we can distinguish the following
cases.
\\
If $c_n(V) = 1$, $V$ is a single $n$-voxel.
If $c_n(V) = 2$,  we have two configurations, depending on the
relative position of the two $n$-voxels $v_1$ and $v_2$. More
precisely, if $v_1$ and $v_2$ are strictly $(n-1)$-adjacent, then
they form an $(n-1)$-block centered on an $(n-1)$-cell that is
bounded by $e$; instead, if they are strictly $(n-2)$-adjacent, they
form an $(n-2)$-gap having $e$ as $(n-2)$-hub.
If $c_n(V) = 3$, by Definition \ref{Def:blocchi and L-block}(2) and
Remark \ref{Not:Numero di nVoxel che compone un blocco e un L
blocco}, the unique possible configuration is given by the $L$-block
centered on $e$.
Finally, if $c_n(V) = 4$,  $V$ coincides with the $(n-2)$-block
centered on $e$.
\end{proof}

\begin{proposition}\label{Pro:12}
Let $v$ be an $n$-voxel and $e$ be one of its $i$-cells, $i=0,\ldots, n-1$.
Then, for any $i < j \le n$, it results:
\[b_j(e) = \frac{c_{i \rightarrow j}c_{j \rightarrow n}}{c_{i \rightarrow n}} .\]
\end{proposition}
\begin{proof}
Let us consider the incidence structure $I=(C_i( v ),C_j( v ),<)$.
By Proposition \ref{Pro:relazione fondamentale delle strutture di
incidenza}, it is $\displaystyle \sum_{a \in C_i( v )} r_a = \sum_{a
\in C_j( v )} k_a$. Evidently, $|C_i( v )| = c_i = c_{i \rightarrow
n}$ and $| C_j( v ) | = c_j  = c_{j \rightarrow n}$, while, for any
$i$-cell $a$ of $C_i( v )$ (respectively $j$-cell $a$ of $C_j( v
)$), $ r_a = b_j(e) $ (respectively $ k_a =  c_{i \rightarrow j}$).
Hence we have $ b_j(e) c_{i \rightarrow n} = c_{i \rightarrow j}
c_{j \rightarrow n} $, from which we get the thesis.
\end{proof}

\begin{cor}\label{Cor:13}
Let $v$ be an $n$-voxel and $e$ be one of its $i$-cell, $i=0,\ldots,
n-1$. Then, for any $i < j \le n$, we have
\[ b_j(e) = \binom{n-i}{j-i} . \]
\end{cor}
\begin{proof}
By Proposition \ref{Pro:12}, it is
\begin{equation*}
    b_j(e) = \frac{c_{i \rightarrow j}c_{j \rightarrow n}}{c_{i \rightarrow n}} = \frac{2^{j-i} \binom{j}{i} 2^{n-j} \binom{n}{j}}  {2^{n-i} \binom{n}{i}} = \frac{j!}  { (j-i)!i!} \cdot \frac{n!}{(n-j)!j!} \cdot \frac{(n-i)!i!}{n!}=\frac{(n-i)!}{(n-j)!(j-i)!} = \binom{n-i}{j-i}.
\end{equation*}
\end{proof}

\begin{lem}\label{Lem:Lemma 1b}
Let $e$ be an $(n-1)$-cell of $\CC_n$. Then the number of $i$-cells
of the $(n-1)$-block centered on $e$ is
\[ c_i(B_{n-1}(e)) =  \frac {3n+i}{2n} c_{i \rightarrow n}.\]
\end{lem}
\begin{proof}
By Remark \ref{Not:Numero di nVoxel che compone un blocco e un L
blocco}, $B_{n-1}(e)$ is composed of two $(n-1)$-adjacent
$n$-voxels. Each of such voxels has exactly $c_{i \rightarrow n}$
$i$-cells, but some of these cells are in common. The number of
these common $i$-cells coincides with the number of $i$-cells of the
center $e$ of the given block. So, we have $\displaystyle c_i(B_{n-1}(e)) = 2 c_{i
\rightarrow n} - c_{i \rightarrow n-1} = 2 \cdot 2 ^{n-i} \binom n i
- 2^{n-1-i} \binom{n-1}{i}= 2 \cdot 2 ^{n-i}  \binom n i - 2
^{n-i-1}  \binom{n}{i} \frac {n-1}{n} =  2 ^{n-i} \binom {n}{i}
\left(  2 - \frac{n-i}{2n} \right) = \frac{3n + i}{2n} c_{i
\rightarrow n }$.
\end{proof}

\begin{lem}\label{Lem:Lemma 2b}
Let $e$ be an $(n-1)$-cell of $\CC_n$. Then the number of free $(n-1)$-cells of the $(n-1)$-block centered on $e$ is:
\[ c^*_{n-1}(B_{n-1}(e)) = 2( 2n -1) . \]
\end{lem}
\begin{proof}
By applying Lemma \ref{Lem:come trovare il numero di n-1celle} to
the digital object $B_{n-1}(e)$, we have $c'_{n-1} + c_{n-1}^* = 2 n
c_n - c'_{n-1}$. But for an $(n-1)$-block it is $c_n = 2$ and
$c'_{n-1} = 1$. Then $c^*_{n-1} = 2 (2n -1)$.
\end{proof}

\begin{proposition}\label{Pro13}
Let $e$ be a free $(n-2)$-cells that belongs to the center of an $(n-1)$-block $B_{n-1}(f)$, then $b_{n-1}(e) = 2$.
\end{proposition}
\begin{proof}
Let us consider the incidence structure $(C_{n-2}(B_{n-1}(f)), C_{n-1}^*(B_{n-1}(f)), <)$. By Lemma \ref{Lem:Lemma 1b}, it is $|C_{n-2}(B_{n-1}(f)) | = c_{n-2} = 2(n-1)(2n - 1)$,
and by Lemma \ref{Lem:Lemma 2b}, we have $|C_{n-1}^*(B_{n-1}(f))| = c_{n-1}^* = 4 n - 2$.
\\
Moreover, by Proposition \ref{Pro:relazione fondamentale delle strutture di incidenza}, it is
\begin{equation*}
    \sum_{a \in C_{n-2}(B_{n-1}(f))} r_a = \sum_{a \in C_{n-1}^*(B_{n-1}(f))} k_a.
\end{equation*}
Since for any $a \in C_{n-1}^*(B_{n-1}(f))$ it is $k_a = c_{n-2 \rightarrow n-1} $, we have \[\sum_{a \in C_{n-1}^*(B_{n-1}(f))} k_a = c^*_{n-1} \cdot c_{n-2 \rightarrow n-1} = (4n - 2) \cdot 2 \cdot (n-1) = 4 (2n - 1)(n-1).\]
Let us consider the sets
\[F = \{ a \in C_{n-2}(B_{n-1}(f)) \colon a < f \} \]
and
\[G = \{ a \in C_{n-2}(B_{n-1}(f)) \colon a \nless f \}.\]
Since $\{F,G\}$ forms a partition of $C_{n-2}(B_{n-1}(f))$, we can
write
\[\sum_{a \in C_{n-2}(B_{n-1}(f))} r_a = \sum_{a \in F} r_a + \sum_{a \in G} r_a.\]
For any $a \in F$, $r_a = b_{n-1}(e)$, and so
\begin{equation*}\label{Equ:Somma F}
    \sum_{a \in F} r_a = |F| b_{n-1}(e) =  c_{n-2 \rightarrow n-1} b_{n-1}(e) = 2 (n-1) b_{n-1}(e).
\end{equation*}
Instead, thanks to Proposition \ref{Pro:12}, for any $a \in G$, we have
\[   r_a = b_{n-1}(e) = \frac{c_{n-2 \rightarrow n - 1}\cdot c_{n-1 \rightarrow n}}{c_{n-2 \rightarrow n}} = 2.\]
Hence, we get that
\begin{equation*}\label{Equ: Somma G}
\sum_{a \in G} r_a = 2 (c_{n-2} - c_{n-2 \rightarrow n-1}) = 2 (2(n-1)(2n-1) - 2(n-1)) = 4(n-1)(2n-1) - 4(n-1).
\end{equation*}
To sum up, we can write $4(n-1)(2n-1) - 4(n-1) + 2 (n-1) b_{n-1}(e) = 4(2n-1)(n-1)$, from which we get the thesis.
\end{proof}

\begin{lem}\label{Lem:Lemma 1c}
Let $e$ be an $(n-2)$-cell of $\CC_n$. Then the number of $i$-cells of the
$L$-block centered on $e$ is:
\[ c_i(L(e)) =  \left(  \frac{2n + i}{n}\right) c_{i \rightarrow n}. \]
\end{lem}
\begin{proof}
By Remark \ref{Not:Numero di nVoxel che compone un blocco e un L
blocco}, $L(e)$ is composed of three $n$-voxels, which are pairwise
$(n-1)$-adjacent in exactly two non-free $(n-1)$-cells. Each of
these three voxels has exactly $c_{i \rightarrow n}$ $i$-cells, but
some of these cells are in common. The number of such common
$i$-cells coincides with the number of $i$-cells of the two non-free
$(n-1)$-cells. So, we have $ \displaystyle c_i(L(e)) = 3 c_{i \rightarrow n} -
2c_{i \rightarrow n-1} = 3 \cdot 2 ^{n-i} \binom n i - 2 \cdot
2^{n-i-1} \binom{n-1}{i}= 3 \cdot 2 ^{n-i} \binom n i - 2^{n-i}
\binom{n}{i} \frac{n-i}{n}= 2^{n-i} \binom{n}{i} \left(3 -
\frac{n-i}{n}\right) = \left(  \frac{2n + i}{2n}\right) c_{i
\rightarrow n} $.
\end{proof}

\begin{lem}\label{Lem: Lemma 2c}
Let $e$ be an $(n-1)$-cell of $\CC_n$. Then
the number of free $(n-1)$-cells of the $L$-block centered on $e$ is:
\[ c^*_{n-1}(L(e)) = 2(3n -2) .\]
\end{lem}
\begin{proof}
By applying Lemma \ref{Lem:come trovare il numero di n-1celle} to
the digital object $L(e)$, we have $c'_{n-1} + c_{n-1}^* = 2 n c_n -
c'_{n-1}$. But for an $L$-block it is $c_n = 3$ and $c'_{n-1} = 2$.
Then $c^*_{n-1} = 2(3n -2)$.
\end{proof}

\begin{proposition}\label{Pro:20}
Let $e$ be a free $(n-2)$-cells which is the center of an $L$-block
$L(e)$. Then $b_{n-1}(e) = 2$.
\end{proposition}
\begin{proof}
Let us consider the incidence structure $(C_{n-2}(L(e)),
C_{n-1}^*(L(e)), <)$. By Lemma \ref{Lem:Lemma 1c}, we have $
|C_{n-2}(L(e)) | =  c_{n-2} = 2(n-1)(3n - 2)$, and by Lemma
\ref{Lem: Lemma 2c}, it is $|C_{n-1}^*(L(e))| = c_{n-1}^* = 2(3 n -
2)$.

By Proposition \ref{Pro:relazione fondamentale delle strutture di incidenza}, it is
\begin{equation}\label{Equ:Re}
    \sum_{a \in C_{n-2}(L(e))} r_a = \sum_{a \in C_{n-1}^*(L(e)} k_a.
\end{equation}

Since for any $a \in C_{n-1}^*(L(e))$ it is $k_a = c_{n-1 \rightarrow n-2} $, we have \[\sum_{a \in C_{n-1}^*(L(e))} k_a = c^*_{n-1} \cdot c_{n-1 \rightarrow n-2} =  2(3n - 2) \cdot 2 \cdot (n-1) = 4 (3n - 2)(n-1).\]
Let us set
$F = \CC_{n-1}'(L(e))$,
and let us consider the sets:
\[A = \{  e  \},\]
\[B = \{ c \in \CC_{n-2}(L(e)) \colon c \nless f, \text{ for some } f \in F \}.\]
\[C = \{ c \in \CC_{n-2}(L(e))\{e\} \colon c < f, \text{ for some } f \in F \}.\]
Let us observe that $|F| = 2$ because the number of $(n-1)$-block of $L(e)$ is $2$. Since $\{A, B, C\}$ forms a partition of $C_{n-2}(L(e))$, it results
\begin{equation}\label{Equ:Re1}
\sum_{a \in C_{n-2}(L(e))} r_a =  r_e + \sum_{a \in B} r_a +  \sum_{a \in C} r_a,
\end{equation}
where, evidently,  $r_e = b_{n-1}(e)$.
\\
Moreover, by Proposition \ref{Pro13}, it is $\sum_{a \in B} r_a = (2
c_{n - 2 \rightarrow n - 1} - 2) \cdot 2 =(2 \cdot 2(n-1) - 2) \cdot
2 = 8(n - 1) - 4$. Finally, by Proposition \ref{Pro:12}, we have
$\sum_{a \in C} r_a = 2 (c_{n-2} - 2 c_{n-2 \rightarrow n-1} + 1) =
2 (2(3n-2)(n-1) - 2 \cdot 2 (n - 1) +1) = 4(3n-2)(n-1) - 8(n-1) + 2
$.
\par
Thus, replacing these results into formulas \ref{Equ:Re1} and \ref{Equ:Re},  we obtain $4(3n-2)(n-1) =b_{n-1}(e) + 8(n-1) - 4 + 4(3n-2)(n-1) - 8 (n - 1) + 2 $, from which we get the thesis.
\end{proof}

\begin{proposition}\label{Lem:1}
Let $D$ be a digital object of $\CC_n$ and $e \in \mathcal H_{n-2}$.
Then $b_{n-1}(e) = 4$.
\end{proposition}
\begin{proof}
Let $v_1$ and $v_2$ be the two $n$-voxels of the $(n-2)$-gap through $e$.
Then the number $b_{n-1}(e)$ of free $(n-1)$-cells of $D$ bounded by $e$ coincides with the maximum number
of $(n-1)$-cells bounded by an $(n-2)$-cell, that is, by Proposition
\ref{Pro:numero massimo di i-celle che sono limitate da una j-cella
fissata}:
\[ b_{n-1}(e) = c_{n-2 \leftarrow n-1} = 2^{(n-1)-(n-2)}\binom{n-(n-2)}{(n-1)-(n-2)} = 4. \]
\end{proof}

\begin{proposition}\label{Lem:2}
Let $D$ be a digital object of $\CC_n$ and $e \in \mathcal N_{n-2}$.
Then $b_{n-1}(e) = 2$.
\end{proposition}

\begin{proof}
Every free $(n-2)$-cell that is not an $(n-2)$-hub is either a
simple cell, or bounds the center of an $(n-1)$-block, or is the
center of an $L$-block. Hence, by Corollary \ref{Cor:13} and
Propositions \ref{Pro13} and \ref{Pro:20}, we get the thesis.
\end{proof}

\begin{proposition}\label{Pro:relazione fondamentale delle strutture di incidenza per gli oggetti digitali}
Let $D$ be a digital $n$-object, and $i < j \le n-1$. Then
\[ \sum_{e \in  bd_i(D)} b_j(e) = c_{i \rightarrow j} c_j^*. \]
\end{proposition}

\begin{proof}
The $i$-border $bd_i(D)$ of $D$ can be considered as an incidence
structure $(V,\mathcal B,\mathcal I)$, where $V= bd_i(D)$, $\mathcal
B = bd_j(D)$, and the incidence relation $\mathcal I$ is the bounding
relation $<$.
\\
In such a structure, the point degree of every vertex $e \in V$
coincides with the number $b_j(e)$ of $j$-cells of $bd(D)$
bounded by $e$. Moreover, the block degree $k_\beta$ of every block
$\mathcal B$ coincides with the maximum number $c_{i\rightarrow j}$
of $i$-cells that bound a $j$-cell. Hence, by Proposition
\ref{Pro:relazione fondamentale delle strutture di incidenza},
$\displaystyle \sum_{e \in bd_i(D)} b_j(e) = \sum_{\beta \in bd_j(D)} c_{i
\rightarrow j} = c_{i \rightarrow j}|bd_j(D)| = c_{i \rightarrow j}
c_j^*$.
\end{proof}

\begin{thm}\label{Teo:number of Gn-2 gaps}
The number of $(n-2)$-gaps of a digital object $D$ of $\CC_n$ is given
by the formula:
\begin{equation}\label{Eq:mia}
    g_{n-2} = (n-1)c_{n-1}^*-c_{n-2}^*.
\end{equation}
\end{thm}
\begin{proof}
Let us consider the sets $\mathcal H_{n-2}$ and $\mathcal N_{n-2}$ of all $(n-2)$-hubs and
$(n-2)$-nubs of $D$, respectively. Evidently $\{ \mathcal H_{n-2}, \mathcal N_{n-2}\}$ is a partition of $bd_{n-2}(D)$.
Moreover, for $i=n-1$ and $j=n-2$, Proposition \ref{Pro:relazione fondamentale delle strutture di incidenza per gli oggetti digitali} give us
\begin{equation*}
    \sum_{e \in bd_{n-2}(D)} b_{n-1}(e) = c_{n-2 \rightarrow n-1} c_{n-1}^* = 2 (n-1)c_{n-1}^*.
\end{equation*}
Since
\[\displaystyle \sum_{e \in bd_{n-2}} b_{n-1}(e)= \sum_{e
\in \mathcal H_{n-2}} b_{n-1}(e) + \sum_{e \in \mathcal N_{n-2}}
b_{n-1}(e) , \]
by Lemmas \ref{Lem:1} and
\ref{Lem:2}, we obtain
\[\sum_{e \in bd_{n-2}} b_{n-1}(e)
= 4 |\mathcal H_{n-2}| + 2 |\mathcal N_{n-2}|
= 4 g_{n-2} + 2 (c^*_{n-2} - g_{n-2})\]
and hence the thesis.
\end{proof}

In \cite{B},  it was proved that the number of
$(n-2)$-gap of a digital $n$-object $D$ can be expressed by
\begin{equation}\label{eq:n-2gap}
    g_{n-2}= -2n(n-1)c_n + 2 (n-1)c_{n-1} - c_{n-2} + \beta_{n-2},
\end{equation}
where $\beta_{n-2}$ is the number of all $(n-2)$-blocks contained in $D$.
\par
Such a formula is equivalent to the expression \eqref{Eq:mia} obtained in
Theorem~\ref{Teo:number of Gn-2 gaps}.
Indeed, we have the following theorem.

\begin{thm}
The formulas
\begin{equation}\label{Equ:Prima}
    g_{n-2}= (n-1)c_{n-1}^*-c_{n-2}^*
\end{equation}
and
\begin{equation}\label{Equ:Seconda}
    g_{n-2}=-2n(n-1)c_{n}+ 2(n-1)c_{n-1}-c_{n-2}+\beta_{n-2}
\end{equation}
are equivalent.
\end{thm}
\begin{proof}
By Lemma \ref{Lem:come trovare il numero di n-1celle}, we have
\[c^*_{n-1}= c_{n-1} - c_{n-1}' = c_{n-1} + c_{n-1} - 2nc_n = 2 c_{n-1} - 2nc_n.\]
Hence, replacing the latter expression in
\eqref{Equ:Prima} ,
we obtain \[g_{n-2}= (n-1)c_{n-1}^*-c_{n-2}^*= 2 (n-1)c_{n-1}- 2
(n-1) c_n - c_{n-2} + c_{n-2}'.\] Finally, since $c_{n-2}'$ is the
number $\beta_{n-2}$ of $(n-2)$-blocks contained in $D$, we get
Formula \eqref{Equ:Seconda}.
\par
Conversely, by Lemma \ref{Lem:come trovare il numero di
n-1celle}, we have $ c_n =\frac{c_{n-1}+ c_{n-1}'}{2n} $. Thus Formula
\eqref{Equ:Seconda} becomes \[g_{n-2} = -2n(n-1) \frac{c_{n-1} +
c_{n-1}'}{2n} + 2(n-1)c_{n-1} + c_{n-2}^* = - (n-1)c_{n-1}'+
(n-1)c_{n-1} + c_{n-2}^* = (n-1)c_{n-1}^* + c_{n-2}^*,\]
that is Formula \eqref{Equ:Prima}. This completes our proof.
\end{proof}

\section{Conclusion and Perspective}\label{sez:concl}
In this paper we have found a new  formula for expressing the number
of $(n-2)$-gaps of a digital $n$-object by means of its free cells.
Unlike the equivalent formula \eqref{eq:n-2gap} given in
\cite{B}, our expression  has the advantage to involve only few
 intrinsic parameters. We hypothesize that such information  could be obtained from some appropriate data structure related to the digital $n$-object. This will be the focus of a forthcoming research.
\par
Another field of investigation could consist in finding a formula, analogous to \eqref{Eq:mia},  which express the number of any $k$-gaps with $0\le k \le n-3$, by means of same basic parameters of the digital $n$-object.

\end{document}